\documentclass[proceedings]{stacs}
\stacsheading{2009}{111--122}{Freiburg}
\firstpageno{111}

\providecommand{\bigo}{\mathcal{O}}
\providecommand{\etal}{{\em et al.}}

\providecommand{\Runs}{\mathtt{Runs}}
\providecommand{\SRuns}{\mathtt{SRuns}}
\providecommand{\HRuns}{\mathtt{HRuns}}
\providecommand{\nRuns}{\rho}
\providecommand{\nSRuns}{\tau}
\providecommand{\SUS}{\mathtt{SUS}}
\providecommand{\nSUS}{{\ensuremath \sigma}}

\begin{document}

\title[Compressed Representations of Permutations]{Compressed
  Representations of Permutations, \\ and Applications}

\author{J. Barbay}{J\'er\'emy Barbay}
\author{G. Navarro}{Gonzalo Navarro}
\address{Dept. of Computer Science (DCC), University of Chile.}
\email{{jbarbay,gnavarro}@dcc.uchile.cl}  

\thanks{Second author partially funded by Fondecyt Grant 1-080019, Chile.}


\begin{abstract}
  We explore various techniques to compress a permutation $\pi$ over
  $n$ integers, taking advantage of ordered subsequences in $\pi$,
  while supporting its application $\pi(i)$ and the application of its
  inverse $\pi^{-1}(i)$ in small time.
  Our compression schemes yield several interesting byproducts, in many
  cases matching, improving or extending the best existing results on 
  applications such as the encoding
  of a permutation in order to support iterated applications $\pi^{k}(i)$ of it,
  of integer functions, and  of inverted lists and suffix arrays.
\end{abstract}

\maketitle

\section{Introduction}
\label{sec:introduction}

Permutations of the integers $[n] = \{1,\ldots,n\}$ are a basic
building block for the succinct encoding of integer
functions~\cite{representingFunctions},
strings~\cite{ANS06,rankSelectOperationsOnLargeAlphabets,NM07,Sad03},
and binary
relations~\cite{succinctIndexesForStringsBinaryRelationsAndMultiLabeledTrees,adaptiveSearchingInSuccinctlyEncodedBinaryRelationsAndTreeStructuredDocumentsTCS},
among others.
A permutation $\pi$ is trivially representable in $n\lceil\lg n\rceil$
bits, which is within $\bigo(n)$ bits of the information theory lower bound 
of $\lg(n!)$ bits.\footnote{In this paper we use the
notations $\lg x = \log_2 x$ and $[x]=\{1,\ldots,x\}$.}
In many interesting applications, efficient computation of both the 
permutation $\pi(i)$ and its inverse $\pi^{-1}(i)$ is required.

The lower bound of $\lg(n!)$ bits yields a lower bound of $\Omega(n\log
n)$ comparisons to sort such a permutation in the comparison model.
Yet, a large body of research has been dedicated to finding better
sorting algorithms which can take advantage of specificities of each
permutation to sort.
Trivial examples are permutations sorted such as the identity, or
containing sorted
blocks~\cite{measuresOfPresortednessAndOptimalSortingAlgorithms}
(e.g. $(\mathit{1},\mathit{3},\mathit{5},\mathit{7},\mathit{9},\mathbf{2},\mathbf{4},\mathbf{6},\mathbf{8},\mathbf{10})$
or
$(\mathit{6},\mathit{7},\mathit{8},\mathit{9},\mathit{10},\mathbf{1},\mathbf{2},\mathbf{3},\mathbf{4},\mathbf{5})$),
or containing sorted
subsequences~\cite{sortingShuffledMonotoneSequences}
(e.g. $(\mathit{1},\mathbf{6},\mathit{2},\mathbf{7},\mathit{3},\mathbf{8},\mathit{4},\mathbf{9},\mathit{5},\mathbf{10})$):
algorithms performing only $\bigo(n)$ comparisons on such
permutations, yet still $\bigo(n\log n)$ comparisons in the worst
case, are achievable and obviously preferable.
Less trivial examples are classes of permutations whose structure
makes them interesting for applications: see Mannila's seminal
paper~\cite{measuresOfPresortednessAndOptimalSortingAlgorithms} and
Estivil-Castro and Wood's review~\cite{estivillcastro92survey} for
more details. 

Each sorting algorithm in the comparison model yields an encoding
scheme for permutations: It suffices to note the result of each
comparison performed to uniquely identify the permutation sorted, and
hence to encode it.
Since an adaptive sorting algorithm performs $o(n\log n)$ comparisons
on many classes of permutations, each adaptive algorithm yields a {\em
  compression scheme} for permutations, at the cost of losing a
constant factor on some other ``bad'' classes of permutations.
We show in Section~\ref{sec:applications} some examples of applications
where only ``easy'' permutations arise.
Yet such compression schemes do not necessarily support in reasonable
time the inverse of the permutation, or even the simple application of
the permutation: this is the topic of our study.
We describe several encodings of permutations so that
on interesting classes of instances the encoding uses $o(n\log n)$ bits while
supporting the operations $\pi(i)$ and $\pi^{-1}(i)$ in time $o(\log n)$.
Later, we apply our compression schemes to
various scenarios, such as the encoding of integer functions, text
indexes, and others, yielding original compression schemes for these
abstract data types.

\section{Previous Work}
\label{sec:previous-work}

\begin{definition}
  The {\em entropy} of a sequence of positive integers $X=\langle n_1,n_2,
  \ldots, n_r\rangle$ adding up to $n$ is $H(X) =
  \sum_{i=1}^r\frac{n_i}{n} \lg \frac{n}{n_i}$.
  By convexity of the logarithm,
    $\frac{r\lg n}{n} \le H(X) \le \lg r$.
\label{def:entrop}
\end{definition}

\paragraph{\bf Succinct Encodings of Sequences}
\label{sec:sequences}

Let $S[1,n]$ be a sequence over an alphabet $[r]$. This includes 
bitmaps when $r=2$ (where, for convenience, the alphabet will be 
$\{0,1\}$). We will make use of succinct representations of $S$
that support operations $rank$ and $select$: 
  $rank_c(S,i)$ gives the number of occurrences of $c$ in $S[1,i]$ and
  $select_c(S,j)$ gives the position in $S$ of the $j$th occurrence of $c$.

For the case $r=2$, $S$ requires $n$ bits of space and $rank$ and $select$ can
be supported in constant time using $\bigo(\frac{n\log\log n}{\log n}) = o(n)$ 
bits on top of $S$ \cite{Mun96,Cla96,Gol06}. 
The extra space is more 
precisely $\bigo(\frac{n\log b}{b} + 2^b\,\textrm{polylog}(b))$ for some
parameter $b$, which is chosen to be, say, $b=\frac{1}{2}\lg n$ to achieve
the given bounds. In this paper, we will sometimes apply the technique over 
sequences of length $\ell =o(n)$ ($n$ will be the length of the permutations). 
Still, we will maintain the value of $b$ as a function of $n$, not 
$\ell$, which ensures that the extra space will be of the form 
$\bigo(\frac{\ell\log\log n}{\log n})$, i.e., it will tend to zero when
divided by $\ell$ as $n$ grows, even if $\ell$ stays constant.
All of our $o()$ terms involving several variables in this paper can be
interpreted in this strong sense: asymptotic in $n$.
Thus we will write the above space simply as $o(\ell)$.

Raman \etal~\cite{RRR02} devised a bitmap representation that takes
$nH_0(S) + o(n)$ bits, while maintaining the constant time for the operations.
Here $H_0(S) = H(\langle n_1,n_2,\ldots,n_r\rangle) \le \lg r$, where $n_c$ is 
the number of occurrences of symbol $c$ in $S$, is the so-called {\em 
zero-order entropy} of $S$. 
For the binary case this simplifies to
$nH_0(S) = m\lg\frac{n}{m} + (n-m)\lg\frac{n}{n-m} = m\lg\frac{n}{m}+\bigo(m)$, 
where $m$ is the number of bits set in $S$.

Grossi \etal~\cite{GGV03} extended the result to larger alphabets using
the so-called {\em wavelet tree}, which decomposes a sequence into several
bitmaps. By representing those bitmaps in plain form, one can represent $S$
using $n\lceil\lg r\rceil(1+o(1))$ bits of space, and answer $S[i]$, as well
as $rank$ and $select$ queries on $S$, in time $\bigo(\log r)$. By, instead,
using Raman \etal's representation for the bitmaps, one achieves
$nH_0(S) + o(n\log r)$ bits of space, and the same times.
Ferragina \etal~\cite{FMMN07} used multiary wavelet trees to maintain the
same compressed space, while improving the times for all the operations to
$\bigo(1+\frac{\log r}{\log\log n})$.

\paragraph{\bf Measures of Disorder in Permutations}
\label{sec:meas-disord-perm}

Various previous studies on the presortedness in
sorting considered in particular the following measures of order on
an input array to be sorted.
Among others, 
Mehlhorn~\cite{mehlhorn} and Guibas~\etal~\cite{guibas} considered the
number of pairs in the wrong order,
Knuth~\cite{theArtOfComputerProgrammingVol3} considered the number of
ascending substrings (runs),
Cook and Kim~\cite{cook}, and later
Mannila~\cite{measuresOfPresortednessAndOptimalSortingAlgorithms}
considered the number of elements which have to be removed to leave a
sorted list,
Mannila~\cite{measuresOfPresortednessAndOptimalSortingAlgorithms}
considered the smallest number of exchanges of arbitrary elements
needed to bring the input into ascending order,
Skiena~\cite{encroachingListsAsAMeasureOfPresortedness} considered the
number of encroaching sequences, obtained by distributing the input
elements into sorted sequences built by additions to both ends,
and Levcopoulos and Petersson~\cite{sortingShuffledMonotoneSequences}
considered Shuffled UpSequences and Shuffled Monotone Sequences.
Estivil-Castro and Wood~\cite{estivillcastro92survey} list them all
and some others.

\section{Compression Techniques}
\label{sec:compr-techn}

We first introduce a compression method that takes advantage of (ascending)
runs in the permutation. Then we consider a stricter variant of the runs,
which allows for further compression in applications when those runs arise,
and in particular allows the representation size to be sublinear in $n$.
Next, we consider a more general type of runs, which need not be
contiguous.

\subsection{Wavelet Tree on Runs}
\label{sec:wavelet-tree-runs}

One of the best known sorting algorithm is merge sort, based on a
simple linear procedure to merge two already sorted arrays, resulting
in a worst case complexity of $\bigo(n\log n)$.
Yet, checking in linear time for \emph{down-step} positions in the
array, where an element is followed by a smaller one, partitions the
original arrays into ascending runs which are already sorted. 
This can speed up the algorithm when the array is partially
sorted~\cite{theArtOfComputerProgrammingVol3}.
We use this same observation to encode permutations.

\begin{definition}
  A \emph{down step} of a permutation $\pi$ over $[n]$ is a position $i$ 
  such that $\pi(i+1)<\pi(i)$.  
  A \emph{run} in a permutation $\pi$ is a maximal range of
  consecutive positions $\{i,\ldots,j\}$ which does not contain any
  down step.
  Let $d_1,d_2, \ldots,d_k$ be the list of consecutive down steps in $\pi$.
  Then the number of runs of $\pi$ is noted $\nRuns = k+1$, 
  and the sequence of the lengths of the runs is noted
  $\Runs = \langle d_1,d_2-d_1, \ldots,d_k-d_{k-1},n+1-d_k \rangle$.
\end{definition}

For example, permutation $(\mathit{1},\mathit{3},\mathit{5},\mathit{7},\mathit{9},\mathbf{2},\mathbf{4},\mathbf{6},\mathbf{8},\mathbf{10})$ contains 
$\nRuns=2$ runs, of lengths $\langle 5,5 \rangle$.
Whereas previous
analyses~\cite{measuresOfPresortednessAndOptimalSortingAlgorithms} of
adaptive sorting algorithms considered only the number $\nRuns$ of runs, we
refine them to consider the distribution $\Runs$ of the sizes of the
runs.

\begin{theorem}
  There is an encoding scheme using at most $n(2+H(\Runs))(1+o(1)) +
  \bigo(\nRuns\log n)$ bits to encode a permutation $\pi$ over $[n]$
  covered by $\nRuns$ runs of lengths $\Runs$.
  It supports $\pi(i)$ and $\pi^{-1}(i)$ in time $\bigo(1+\log\nRuns)$
  for any value of $i\in[n]$.
  If $i$ is chosen uniformly at random in $[n]$ then the average time
  is $\bigo(1+H(\Runs))$.
\label{thm:main}
\end{theorem}

\begin{proof}
The Hu-Tucker algorithm \cite{HT71} (see also Knuth
\cite[p.~446]{theArtOfComputerProgrammingVol3}) produces in
$\bigo(\nRuns\log\nRuns)$ time a prefix-free code from a sequence of 
frequencies $X=\langle n_1,n_2, \ldots,n_\nRuns\rangle$ adding up to $n$, 
so that (1) the $i$-th lexicographically smallest code is that for frequency 
$n_i$, and (2) if $\ell_i$ is the bit length of the code assigned to the 
$i$-th sequence element, then $L=\sum \ell_i n_i$ is minimal and moreover
$L < n(2+H(X))$ \cite[p.~446, Eq.~(27)]{theArtOfComputerProgrammingVol3}.

We first determine $\Runs$ in $\bigo(n)$ time, and then apply the Hu-Tucker 
algorithm to $\Runs$. We arrange the set of codes produced in a binary
trie (equivalent to a Huffman tree \cite{Huf52}), where each leaf corresponds
to a run and points to its two endpoints in $\pi$. Because of property (1), 
reading the leaves left-to-right yields the runs also in left-to-right order. 
Now we convert this trie into a wavelet-tree-like structure \cite{GGV03} without altering its 
shape, as follows. Starting from the root, first process recursively each 
child. For the leaves do nothing. Once both children of an internal node have
been processed, the invariant is that they point to the contiguous area in 
$\pi$ covering all their leaves, and that this area of $\pi$ has already been 
sorted. Now we merge the areas of the two children in time proportional to the 
new area created (which, again, is contiguous in $\pi$ because of property (1)).
As we do the merging, each time we take an element from the left child we 
append a 0 bit to a bitmap we create for the node, and a 1 bit when we take an 
element from the right list.

When we finish, we have the following facts: (1) $\pi$ has been sorted,
(2) the time for sorting has been $\bigo(n+\nRuns\log\nRuns)$ plus the total 
number of bits appended to all bitmaps, (3) each of the $n_i$ elements of leaf 
$i$ (at depth $\ell_i$) has been merged $\ell_i$ times, contributing $\ell_i$ 
bits to the bitmaps of its ancestors, and thus the total number of bits is 
$\sum n_i \ell_i$.

Therefore, the total number of bits in the Hu-Tucker-shaped wavelet tree is at 
most $n(2+H(\Runs))$. To this we must add the $\bigo(\nRuns\log n)$ bits of the
tree pointers. We preprocess all the bitmaps for $rank$ and $select$ queries
so as to spend $o(n(2+H(\Runs))$ extra bits (\S\ref{sec:sequences}). 

To compute $\pi^{-1}(i)$ we start at offset $i$ at the root bitmap $B$, with
position $p \leftarrow 0$, and bitmap size $s \leftarrow n$. If $B[i] = 0$ we 
go down to the left child with $i \leftarrow rank_0(B,i)$ and $s \leftarrow
rank_0(B,s)$. Otherwise we go down to the right child with 
$i \leftarrow rank_1(B,i)$, $p \leftarrow p + rank_0(B,s)$, and $s \leftarrow
rank_1(B,s)$. When we reach a leaf, the answer is $p+i$.

To compute $\pi(i)$ we do the reverse process, but we must first determine the 
leaf $v$ and offset $j$ within $v$ corresponding to position $i$: We start at 
the root bitmap $B$, with bitmap size $s \leftarrow n$ and position 
$j \leftarrow i$. If $rank_0(B,s) \ge j$ we go down to the left child with 
$s \leftarrow rank_0(B,s)$. Otherwise we go down to the right child with 
$j \leftarrow j-rank_0(B,s)$ and $s \leftarrow rank_1(B,s)$. We eventually reach
leaf $v$, and the offset within $v$ is $j$. We now start an upward traversal
using the nodes that are already in the recursion stack (those will be limited
to $\bigo(\log\nRuns)$ soon). If $v$ is a left child of its parent $u$, then we 
set $j \leftarrow select_0(B,j)$, else we set $j \leftarrow select_1(B,j)$, 
where $B$ is the bitmap of $u$. Then we set $v \leftarrow u$ until reaching 
the root, where $j = \pi(i)$.

In both cases the time is $\bigo(\ell)$, where $\ell$ is the depth of the leaf 
arrived at.
If $i$ is chosen uniformly at random in $[n]$, then the average
cost is $\frac{1}{n}\sum n_i\ell_i = \bigo(1+H(\Runs))$. However, the worst case
can be $\bigo(\nRuns)$ in a fully skewed tree. We can ensure $\ell = 
\bigo(\log \nRuns)$ in the worst case while maintaining the average case by 
slightly rebalancing the Hu-Tucker tree: If there exist nodes at depth 
$\ell=4\lg\nRuns$, we rebalance their subtrees, so as to guarantee maximum 
depth $5\lg\nRuns$. This affects only marginally the size of the structure. 
A node at depth $\ell$ cannot add up to a frequency higher than 
$n/2^{\lfloor \ell/2 \rfloor} \le 2n/\nRuns^2$ (see next paragraph). Added
over all the possible $\nRuns$ nodes we have a total frequency of $2n/\nRuns$.
Therefore, by rebalancing those subtrees we add at 
most $\frac{2n\lg\nRuns}{\nRuns}$ bits. This is $o(n)$ if $\nRuns=\omega(1)$,
and otherwise the cost was $\bigo(\nRuns)=\bigo(1)$ anyway. For the same 
reasons the average time stays $\bigo(1+H(\Runs))$ as it increases at most
by $\bigo(\frac{\log\nRuns}{\nRuns}) = \bigo(1)$.

The bound on the frequency at depth $\ell$ is proved as follows. Consider the 
node $v$ at depth $\ell$, and its grandparent $u$. Then the uncle of $v$ 
cannot have smaller frequency than $v$. 
Otherwise we could improve the already
optimal Hu-Tucker tree by executing either a single (if $v$ is left-left or 
right-right grandchild of $u$) or double (if $v$ is left-right or right-left 
grandchild of $u$) AVL-like rotation that decreases the depth of $v$ by 1 and 
increases that of the uncle of $v$ by 1. Thus the overall frequency at 
least doubles whenever we go up two nodes from $v$, and this holds recursively.
Thus the weight of $v$ is at most $n/2^{\lfloor \ell/2 \rfloor}$. 
\end{proof}

The general result of the theorem can be simplified when the distribution 
$\Runs$ is not particularly favorable.

\begin{corollary}
  There is an encoding scheme using at most
  $n\lceil\lg\nRuns\rceil(1+o(1)) +\bigo(\log n)$ bits to encode a
  permutation $\pi$ over $[n]$ with a set of $\nRuns$ runs.
  It supports $\pi(i)$ and $\pi^{-1}(i)$ in time $\bigo(1+\log\nRuns)$
  for any value of $i\in[n]$.
\label{cor:mainbal}
\end{corollary}

As a corollary, we obtain a new proof of a well-known result on adaptive
algorithms telling that one can sort in time $O(n(1+\log\nRuns))$~%
\cite{measuresOfPresortednessAndOptimalSortingAlgorithms},
now refined to consider the entropy of the partition and not only its size.

\begin{corollary}
  We can sort an array of length $n$ covered by $\nRuns$ runs of
  lengths $\Runs$ in time $\bigo(n(1+H(\Runs)))$, which is worst-case
  optimal in the comparison model among all permutations with
  $\nRuns$ runs of lengths $\Runs$ so that $\nRuns\log n =
  o(nH(\Runs))$.
\label{cor:mainsort}
\end{corollary}

\subsection{Stricter Runs}
\label{sec:stricter}

Some classes of permutations can be covered by a small number of
runs of a stricter type.
We present an encoding scheme which uses $o(n)$ bits for encoding the
permutations from those classes, and still $\bigo(n\lg n)$ bits for
all others.

\begin{definition}
  A \emph{strict run} in a permutation $\pi$ is a maximal range of positions
  satisfying $\pi(i+k)=\pi(i)+k$. The {\em head} of such run is its first
  position.
  The number of strict runs of $\pi$ is noted $\nSRuns$, 
  and the sequence of the lengths of the strict runs is noted
  $\SRuns$. We will call $\HRuns$ the sequence of run lengths of the 
  sequence formed by the strict run heads of $\pi$.
\end{definition}

For example, permutation
$(\mathit{6},\mathit{7},\mathit{8},\mathit{9},\mathit{10},\mathbf{1},\mathbf{2},\mathbf{3},\mathbf{4},\mathbf{5})$ contains $\nSRuns=2$ strict runs, of lengths
$\SRuns = \langle 5,5 \rangle$. The run heads are $\langle \mathit{6},
\mathbf{1} \rangle$, and contain 2 runs, of lengths $\HRuns =
\langle 1,1 \rangle$. Instead, 
$(\mathit{1},\mathit{3},\mathit{5},\mathit{7},\mathit{9},\mathbf{2},\mathbf{4},\mathbf{6},\mathbf{8},\mathbf{10})$ contains $\nSRuns=10$ strict runs, all of length 1.

\begin{theorem}
  There is an encoding scheme using at most $\nSRuns H(\HRuns)(1+o(1)) + 
  2\nSRuns\lg\frac{n}{\nSRuns} + o(n) + \bigo(\nSRuns + \nRuns\log\nSRuns)$ 
  bits to encode a permutation $\pi$ over $[n]$ covered by $\nSRuns$ strict 
  runs and by $\nRuns\le\nSRuns$ runs, and with $\HRuns$ being the $\nRuns$ 
  run lengths in the permutation of strict run heads. 
  It supports $\pi(i)$ and $\pi^{-1}(i)$ in 
  time $\bigo(1+\log\nRuns)$ for any value of $i\in[n]$. If $i$ is chosen 
  uniformly at random in $[n]$ then the average time is $\bigo(1+H(\HRuns))$.
\label{thm:strict}
\end{theorem}

\begin{proof}
We first set up a bitmap $R$ marking with a 1 bit the beginning of the strict
runs. Set up a second bitmap $R^{inv}$ such that $R^{inv}[i] = R[\pi^{-1}(i)]$.
Now we create a new permutation $\pi'$ of $[\nSRuns]$ which collapses the 
strict runs of $\pi$, $\pi'(i) = rank_1(R^{inv},\pi(select_1(R,i)))$.
All this takes $\bigo(n)$ time and the bitmaps take 
$2\nSRuns\lg\frac{n}{\nSRuns} + \bigo(\nSRuns) + o(n)$ bits using 
Raman \etal's technique, where $rank$ and $select$ are solved in constant time
(\S\ref{sec:sequences}).

Now build the structure of Thm.~\ref{thm:main} for $\pi'$. The number of down
steps in $\pi$ is the same as for the sequence of strict run heads in $\pi$,
and in turn the same as the down steps in $\pi'$. So the number of
runs in $\pi'$ is also $\nRuns$ and their lengths are $\HRuns$.
Thus we get at most $\nSRuns(2+H(\HRuns))(1+o(1)) + \bigo(\nRuns\log \nSRuns)$
bits to encode $\pi'$, and can compute $\pi'$ and its inverse in 
$\bigo(1+\log \nRuns)$ worst case and $\bigo(1+H(\HRuns))$ average time.

To compute $\pi(i)$, we find $i' \leftarrow rank_1(R,i)$ and then compute 
$j' \leftarrow \pi'(i')$. The final answer is 
$select_1(R^{inv},j') + i-select_1(R,i')$.
To compute $\pi^{-1}(i)$, we find $i' \leftarrow rank_1(R^{inv},i)$ and then
compute $j' \leftarrow (\pi')^{-1}(i')$. The final answer is 
$select_1(R,j') + i-select_1(R^{inv},i')$. This adds only constant time on top
of that to compute $\pi'$ and its inverse.
\end{proof}

Once again, we might simplify the results when the distribution $\HRuns$ is
not particularly favorable, and we also obtain
interesting algorithmic results on sorting.

\begin{corollary}
  There is an encoding scheme using at most 
  $\nSRuns\lceil\lg\nRuns\rceil (1+o(1)) + 
  2\nSRuns\lg\frac{n}{\nSRuns} + \bigo(\nSRuns) + o(n)$ bits to encode
  a permutation $\pi$ over $[n]$ covered by $\nSRuns$ 
  strict runs and by $\nRuns\le\nSRuns$ runs. It supports 
  $\pi(i)$ and $\pi^{-1}(i)$ in time $\bigo(1+\log\nRuns)$ 
  for any value of $i\in[n]$. 
\label{cor:strictbal}
\end{corollary}

\begin{corollary}
We can sort a permutation of $[n]$, covered by $\nSRuns$ strict runs and
by $\nRuns$ runs, and $\HRuns$ being the run lengths of the strict run heads,
in time $\bigo(n + \nSRuns H(\HRuns)) = \bigo(n+\nSRuns\log\nRuns)$, which is
worst-case optimal, in the comparison model, among all permutations sharing these
$\nRuns$, $\nSRuns$, and $\HRuns$ values, such that $\nRuns\log \nSRuns =
o(\nSRuns H(\HRuns))$.
\label{cor:strictsort}
\end{corollary}

\subsection{Shuffled Sequences}
\label{sec:shuffled-upsequences}

Levcopoulos and Petersson~\cite{sortingShuffledMonotoneSequences}
introduced the more sophisticated concept of partitions formed by
interleaved runs, such as \emph{Shuffled UpSequences} (SUS).
We discuss here the advantage of considering permutations
formed by shuffling a small number of runs.

\begin{definition}
  A decomposition of a permutation $\pi$ over $[n]$ into
  \emph{Shuffled UpSequences} is a set of, not necessarily consecutive,
  subsequences of increasing numbers that have to be removed from $\pi$ 
  in order to reduce it to the empty sequence.
  The minimum number of shuffled upsequences in such a decomposition
  of $\pi$ is noted $\nSUS$, and the sequence of the lengths of the
  involved shuffled upsequences, in arbitrary order, is noted $\SUS$.
\end{definition}

For example, permutation
$(\mathit{1},\mathbf{6},\mathit{2},\mathbf{7},\mathit{3},\mathbf{8},\mathit{4},\mathbf{9},\mathit{5},\mathbf{10})$ contains $\nSUS=2$ shuffled upsequences
of lengths $\SUS=\langle 5,5\rangle$, but $\nRuns=5$ runs, all of length 2.
Whereas the decomposition of a permutation into runs or strict runs
can be computed in linear time, the decomposition into shuffled
upsequences requires a bit more time.
Fredman~\cite{onComputingTheLengthOfLongestIncreasingSubsequences}
gave an algorithm to compute the size of an optimal partition,
claiming a worst case complexity of $\bigo(n\log n)$.
In fact his algorithm is adaptive and takes $\bigo(n(1+\log\nSUS))$ time.
We give here a variant of his algorithm which computes the partition
itself within the same complexity, 
and we achieve even better time on favorable sequences $\SUS$.

\begin{lemma}
  \label{lem:partitionInSUS}
  Given a permutation $\pi$ over $[n]$ covered by $\nSUS$ shuffled
  upsequences of lengths $\SUS$, there is an algorithm finding such a 
  partition in time $\bigo(n(1+H(\SUS)))$.
\end{lemma}
\begin{proof}
  Initialize a sequence $S_1=(\pi(1))$, and a splay tree $T$ \cite{ST85}
  with the node $(S_1)$, ordered by the rightmost value of
  the sequence contained by each node.
  For each further element $\pi(i)$, search for the sequence with the
  maximum ending point smaller than $\pi(i)$.
  If any, add $\pi(i)$ to this sequence, otherwise create a new
  sequence and add it to $T$.
  Fredman~\cite{onComputingTheLengthOfLongestIncreasingSubsequences}
  already proved that this algorithm computes an optimal
  partition. 
  The adaptive complexity results from the mere observation that the
  splay tree (a simple sorted array in Fredman's proof)
  contains at most $\nSUS$ elements, and that the node corresponding
  to a subsequence is accessed once per element in it. Hence the
  total access time is $\bigo(n(1+H(\SUS)))$ \cite[Thm.~2]{ST85}. 
  %
\end{proof}

The complete description of the permutation requires to encode the
computation of both the partitioning algorithm and the sorting one,
and this time the encoding cost of partitioning is as important as
that of merging.

\begin{theorem} \label{thm:SUS} There is an encoding scheme using at
  most $2n(1+H(\SUS)) + o(n\log\nSUS) + \bigo(\nSUS\log n)$ bits to
  encode a permutation $\pi$ over $[n]$ covered by $\nSUS$ shuffled
  upsequences of lengths $\SUS$.
  It supports the operations $\pi(i)$ and $\pi^{-1}(i)$ in time
  $\bigo(1+\log\nSUS)$ for any value of $i\in[n]$.
  If $i$ is chosen uniformly at random in $[n]$ the average time
  is $\bigo(1+H(\SUS)+\frac{\log\nSUS}{\log\log n})$.
\end{theorem}

\begin{proof}
  Partition the permutation $\pi$ into $\nSUS$ shuffled 
  upsequences using Lemma~\ref{lem:partitionInSUS}, resulting in a
  string $S$ of length $n$ over alphabet $[\nSUS]$ which indicates for
  each element of the permutation $\pi$ the label of the upsequence it
  belongs to.
  Encode $S$ with a wavelet tree using Raman~\etal's compression for the 
  bitmaps, so as to achieve $nH(\SUS)+o(n\log\nSUS)$ bits of space and
  support retrieval of any $S[i]$, as well as symbol $rank$ and $select$ on 
  $S$, in time $\bigo(1+\log\nSUS)$ (\S\ref{sec:sequences}).
  Store also an array $A[1,\nSUS]$ so that $A[\ell]$ is the accumulated
  length of all the upsequences with label less than $\ell$. Array $A$ requires
  $\bigo(\nSUS\log n)$ bits.
  Finally, consider the permutation $\pi'$ formed by the upsequences
  taken in label order: $\pi'$ has at most $\nSUS$ runs and hence
  can be encoded using $n(2+H(\SUS))(1+o(1))+\bigo(\nSUS\log n)$ bits
  using Thm.~\ref{thm:main}, as $\SUS$ in $\pi$ corresponds to $\Runs$ in
  $\pi'$. This supports $\pi'(i)$ and $\pi'^{-1}(i)$ in time 
  $\bigo(1+\log\nSUS)$.
  
  Now $\pi(i) = \pi'(A[S[i]]+rank_{S[i]}(S,i))$ can be computed in time
  $\bigo(1+\log\nSUS)$. Similarly, 
  $\pi^{-1}(i) = select_\ell(S,(\pi')^{-1}(i)-A[\ell])$, where $\ell$ is
  such that $A[\ell] < (\pi')^{-1}(i) \le A[\ell+1]$, can also be computed in
  $\bigo(1+\log\nSUS)$ time.
  Thus the whole structure uses $2n(1+H(\SUS)) + o(n\log \nSUS) +
  \bigo(\nSUS\log n)$ bits and
  supports $\pi(i)$ and $\pi^{-1}(i)$ in time $\bigo(1+\log\nSUS)$.  

  The obstacles to achieve the claimed average time are the operations on the
  wavelet tree of $S$, and the binary search in $A$. The former can be reduced
  to $\bigo(1+\frac{\log\nSUS}{\log\log n})$ by using the improved wavelet
  tree representation by Ferragina \etal~(\S\ref{sec:sequences}). 
  The latter is reduced
  to constant time by representing $A$ with a bitmap $A'[1,n]$ with the bits 
  set at the values $A[\ell]+1$, so that $A[\ell] = select_1(A',\ell)-1$, and
  the binary search is replaced by $\ell = rank_1(A',(\pi')^{-1}(i))$. With 
  Raman \etal's structure (\S\ref{sec:sequences}), $A'$ needs 
  $\bigo(\nSUS\log \frac{n}{\nSUS})$ bits and operates in constant time.
\end{proof}

Again, we might prefer a simplified result when $\SUS$ has no interesting
distribution, and we also achieve an improved result on sorting, better than
the known $\bigo(n(1+\log\nSUS))$.

\begin{corollary} \label{cor:SUSbal} 
  There is an encoding scheme using at most $2n\lg\nSUS(1+o(1)) +
  \nSUS\lg\frac{n}{\nSUS} + \bigo(\nSUS)$ bits to encode a permutation
  $\pi$ over $[n]$ covered by $\nSUS$ shuffled upsequences.
  It supports the operations $\pi(i)$ and $\pi^{-1}(i)$ in time
  $\bigo(1+\log\nSUS)$ for any value of $i\in[n]$.
\end{corollary}

\begin{corollary}
  \label{cor:sortSUS} 
  We can sort an array of length $n$, covered by $\nSUS$ shuffled
  upsequences of lenghts $\SUS$, in time $\bigo(n(1+H(\SUS)))$, which
  is worst-case optimal, in the comparison model, among all
  permutations decomposable into $\nSUS$ shuffled upsequences of lenghts
  $\SUS$ such that $\nSUS\log n = o(nH(\SUS))$.
\end{corollary}


\section{Applications}
\label{sec:applications}

\subsection{Inverted Indexes}
\label{sec:invfiles}

Consider a full-text inverted index which gives the word positions of any 
word in a text. This is a popular data structure for natural language text
retrieval \cite{BYRN99,WMB99}, as it permits for example solving phrase queries
without accessing the text. For each different text word, an increasing list of
its text positions is stored.

Let $n$ be the total number of words in a text collection $T[1,n]$ and $\nRuns$
the vocabulary size (i.e., number of different words). An uncompressed inverted 
index requires $(\nRuns+n)\lceil\lg n\rceil$ bits. It has been shown 
\cite{MN07} that, by $\delta$-encoding the differences between consecutive 
entries in the inverted lists, the total space reduces to $nH_0(T) +
\nRuns\lceil\lg n\rceil$, where $H_0(T)$ is the zero-order entropy of the text
if seen as a sequence of words (\S\ref{sec:sequences}). We note that 
the empirical law by Heaps \cite{Hea78}, well accepted in Information Retrieval,
establishes that $\nRuns$ is small: $\nRuns = \bigo(n^\beta)$ for some constant
$0<\beta<1$ depending on the text type.

Several successful methods to compress natural language text take words as
symbols and use zero-order encoding, and thus the size they can achieve is
lower bounded by $nH_0(T)$ \cite{MNZBY00}. If we add the differentially
encoded inverted index in order to be able of searching the compressed text, 
the total space is at least $2nH_0(T)$. 

Now, the concatenation of the $\nRuns$ inverted lists can be 
seen as a permutation of $[n]$ with $\nRuns$ runs, and therefore 
Thm.~\ref{thm:main} lets us encode it in 
$n(2+H_0(T))(1+o(1)) + \bigo(\nRuns\log n)$ bits. Within the same space we can 
add $\nRuns$ numbers telling where the runs begin, in an array $V[1,\nRuns]$. 
Now, in order to retrieve the list of the $i$-th word, we simply obtain
$\pi(V[i]), \pi(V[i]+1), \ldots, \pi(V[i+1]-1)$, each in $\bigo(1+\log\nRuns)$ 
time. Moreover we can extract any random position from a list, which enables
binary-search-based strategies for list intersection \cite{BY04,ST07,CM07}.
In addition, we can also obtain a text passage from the (inverse) permutation: 
To find out $T[j]$, $\pi^{-1}(j)$ gives its position in the inverted lists, and
a binary search on $V$ finds the interval $V[i] \le \pi^{-1}(j) < V[i+1]$, to
output that $T[j] = i$th word, in $\bigo(1+\log \nRuns)$ time.

This result is very interesting, as it constitutes a true word-based 
{\em self-index} \cite{NM07} (i.e., a compressed text index that contains 
the text). Similar results have been recently obtained with rather different 
methods \cite{BFLN08,CN08}. The cleanest one is to
build a wavelet tree over $T$ with compression \cite{FMMN07}, which achieves 
$nH_0(T)+ o(n\log\nRuns)+\bigo(\nRuns\log n)$ bits of space, and permits 
obtaining $T[i]$, as well as extracting the $j$th element of the inverted list 
of the $i$th word with $select_i(T,j)$, all in time
$\bigo(1+\frac{\log\nRuns}{\log\log n})$. 

Yet, one advantage of our approach is that the extraction of $\ell$ consecutive
entries $\pi^{-1}([i,i'])$ takes $\bigo(\ell(1+\log\frac{\nRuns}{\ell}))$ time 
if we do the process for all the entries as a block: 
Start at range $[i,i']$ at the root bitmap 
$B$, with position $p \leftarrow 0$, and bitmap size $s \leftarrow n$. Go down
to both left and right children: to the left with $[i,i'] \leftarrow
[rank_0(B,i),rank_0(B,i')]$, same $p$, and $s \leftarrow rank_0(B,s)$; to the
right with $[i,i'] \leftarrow [rank_1(B,i),rank_1(B,i')]$, 
$p \leftarrow p + rank_0(B,s)$, and $s \leftarrow rank_1(B,s)$. Stop when the
range $[i,i']$ becomes empty or when we reach a leaf, in which case report all
answers $p+k$, $i \le k \le i'$. By representing the inverted list as
$\pi^{-1}$, we can extract long inverted lists faster than the existing
methods.

\begin{corollary} \label{cor:invfile}
There exists a representation for a text $T[1,n]$ of integers in $[1,\nRuns]$
(regarded as word identifiers), with zero-order entropy $H_0$, that takes 
$n(2+H_0)(1+o(1)) + \bigo(\nRuns\log n)$ bits of space, and can retrieve the 
text position of the $j$th occurrence of the $i$th text word, as well as the 
value $T[j]$, in $\bigo(1+\log \nRuns)$ time. It can also retrieve any range of
$\ell$ successive occurrences of the $i$th text word in time $\bigo(\ell (1+\log
\frac{\nRuns}{\ell}))$.
\end{corollary}

We could, instead, represent the inverted list as $\pi$, so as to
extract long text passages efficiently, but the wavelet tree
representation can achieve the same result. Another interesting
functionality that both representations share, and which is useful for
other list intersection algorithms
\cite{fasterAdaptiveSetIntersectionsForTextSearching,adaptiveSearchingInSuccinctlyEncodedBinaryRelationsAndTreeStructuredDocumentsTCS},
is that to obtain the first entry of a list which is larger than
$x$. This is done with $rank$ and $select$ on the wavelet tree
representation. In our permutation representation, we can also achieve
it in $\bigo(1+\log\nRuns)$ time by finding out the position of a
number $x$ within a given run. The algorithm is similar to those in
Thm.~\ref{thm:main} that descend to a leaf while maintaining the
offset within the node, except that the decision on whether to descend
left or right depends on the leaf we want to arrive at and not on the
bitmap content (this is actually the algorithm to compute $rank$ on
binary wavelet trees \cite{NM07}).

Finally, we note that our inverted index data structure supports in
small time all the operations required to solve conjunctive queries on
binary relations.

\subsection{Suffix Arrays}
\label{sec:suffix-arrays}

Suffix arrays are used to index texts that cannot be handled with inverted
lists. Given a text $T[1,n]$ of $n$ symbols over an alphabet of size $\nRuns$,
the {\em suffix} array $A[1,n]$ is a permutation of $[n]$ so that $T[A[i],n]$
is lexicographically smaller than $T[A[i+1],n]$. As suffix arrays take much 
space, several compressed data structures have been developed for them 
\cite{NM07}. One of interest for us is the {\em Compressed Suffix Array (CSA)}
of Sadakane \cite{Sad03}. It builds over a permutation $\Psi$ of $[n]$, which
satisfies $A[\Psi[i]] = (A[i]~\textrm{mod}~n) +1$ (and thus lets us move
virtually one position forward in the text) \cite{GV06}. It turns out that, 
using
just $\Psi$ and $\bigo(\nRuns\log n)$ extra bits, one can $(i)$ {\em count} the
number of times a pattern $P[1,m]$ occurs in $T$ using $\bigo(m\log n)$
applications of $\Psi$; $(ii)$ {\em locate} any such occurrence using 
$\bigo(s)$ applications of $\Psi$, by spending $\bigo(\frac{n\log n}{s})$
extra bits of space; and $(iii)$ {\em extract} a text substring $T[l,r]$ using
at most $s+r-l$ applications of $\Psi$. Hence this is another self-index, and
its main burden of space is that to represent permutation $\Psi$.

Sadakane shows that $\Psi$ has at most $\nRuns$ runs, and gives a 
representation that accesses $\Psi[i]$ in constant time by using
$nH_0(T) + \bigo(n\log\log\nRuns)$ bits of space. It was shown later \cite{NM07}
that the space is actually $nH_k(T) + \bigo(n\log\log\nRuns)$ bits, for any 
$k \le \alpha \log_\nRuns n$ and constant $0<\alpha<1$. Here $H_k(T) \le
H_0(T)$ is the $k$th order empirical entropy of $T$ \cite{Man01}.

With Thm.~\ref{thm:main} we can encode $\Psi$ using $n(2+H_0(T))(1+o(1)) + 
\bigo(\nRuns\log n)$ bits of space, whose extra terms aside from entropy
are better than Sadakane's. Those extra terms can be very significant in
practice. The price is that the time to access $\Psi$ is $\bigo(1+\log \nRuns)$
instead of constant. On the other hand, an interesting extra functionality is
that to compute $\Psi^{-1}$, which lets us move (virtually) one position
backward in $T$. This allows, for example, displaying the text context around
an occurrence without having to spend any extra space. Still, although
interesting, the result is not competitive with recent developments 
\cite{FMMN07,MNspire07}.

An interesting point is that $\Psi$ contains $\nSRuns \le 
\min(n,nH_k(T)+\nRuns^k)$ strict runs, for any $k$ \cite{MN05}. Therefore, 
Cor.~\ref{cor:strictbal} lets us represent it using
$\nSRuns\lceil\lg\nRuns\rceil (1+o(1)) + 2\nSRuns\lg\frac{n}{\nSRuns} + 
\bigo(\nSRuns) + o(n)$ bits of space. For $k$ limited as above, this is 
at most $nH_k(T) (\lg\nRuns + 2\lg\frac{1}{H_k(T)} + \bigo(1))+o(n\log\nRuns)$
bits, which is similar to the space achieved by another self-index 
\cite{MN05,MNSV08}, yet again it is slightly superseded by its time 
performance.

\subsection{Iterated Permutation}
\label{sec:iterated-permutation}

Munro~\etal~\cite{Munro03} described how to represent a permutation
$\pi$ as the concatenation of its cycles, completed by a bitvector of
$n$ bits coding the lengths of the cycles.
As the cycle representation is itself a permutation of $[n]$, we can use any
of the permutation encodings described in
\S\ref{sec:compr-techn} to encode it, adding the
binary vector encoding the lengths of the cycles.
It is important to note that, for a specific permutation $\pi$, the
difficulty to compress its cycle encoding $\pi'$ is not the same as
the difficulty to encode the original permutation $\pi$.

Given a permutation $\pi$ with $c$ cycles of lengths
$\langle n_1,\ldots,n_c\rangle$, there are several ways to encode it as a
permutation $\pi'$, depending on the starting point of each cycle
($\Pi_{i\in[c]} n_i$ choices) and the order of the cycles in the encoding ($c!$
choices).
As a consequence, each permutation $\pi$ with $c$ cycles of lengths
$\langle n_1,\ldots,n_c\rangle$ can be encoded by any of the 
$\Pi_{i\in[c]} i\times n_i$ corresponding permutations.

\begin{corollary}
  Any of the encodings from Theorems~\ref{thm:main}, \ref{thm:strict}
  and~\ref{thm:SUS} can be combined with an additional cost of at most
  $n+o(n)$ bits to encode a permutation $\pi$ over $[n]$ composed of
  $c$ cycles of lengths $\langle n_1,\ldots,n_c\rangle$ to support the
  operation $\pi^{k}(i)$ for any value of $k\in\mathbb{Z}$, in time
  and space function of the order in the permutation encoding of the
  cycles of $\pi$.
\end{corollary}

The space ``wasted'' by such a permutation representation of the
cycles of $\pi$ is $\sum\lg n_i + c\lg c$ bits.
To recover some of this space, one can define a canonical cycle
encoding by starting the encoding of each cycle with its smallest
value, and by ordering the cycles in order of their starting point.
This canonical encoding always starts with a $1$ and creates at least
one shuffled upsequence of length $c$: it can be compressed as a
permutation over $[n-1]$ with at least one shuffled upsequence of
length $c+1$ through Thm~\ref{thm:SUS}.

\subsection{Integer Functions}
\label{sec:integer-functions}

Munro and Rao~\cite{representingFunctions} extended the results on
permutations to arbitrary functions from $[n]$ to $[n]$, and to their
iterated application $f^k(i)$, the function iterated $k$ times
starting at $i$.
Their encoding is based on the decomposition of the function into a
bijective part, represented as a permutation, and an injective part,
represented as a forest of trees whose roots are elements of the
permutation: the summary of the concept is that an integer function is
just a ``hairy permutation''.
Combining the representation of permutations from \cite{Munro03} with
any representation of trees supporting the level-ancestor operator
and an iterator of the descendants at a given level yields a
representation of an integer function $f$ using $(1+\varepsilon)n\lg
n+\bigo(1)$ bits to support $f^k(i)$ in $\bigo(1+|f^k(i)|)$ time, for
any fixed $\varepsilon$, integer $k\in\mathbb{Z}$ and $i\in[n]$.

Janssen~\etal~\cite{ultraSuccinctRepresentationOfORderedTrees}
defined the \emph{degree entropy} of an ordered tree $T$ with $n$
nodes, having $n_i$ nodes with $i$ children, as
$H^*(T)=H(\langle n_1,n_2,\ldots\rangle)$, and proposed a succinct
data structure for $T$ using $nH^*(T)+\bigo(n(\lg\lg n)^2/\lg n)$ bits
to encode the tree and support, among others, the level-ancestor
operator.
Obviously, the definition and encoding can be generalized to a forest
of $k$ trees by simply adding one node whose $k$ children are the
roots of the $k$ trees.

Encoding the injective parts of the function using
Janssen~\etal's~\cite{ultraSuccinctRepresentationOfORderedTrees}
succinct encoding, and the bijective parts of the function using one
of our permutation encodings, yields a compressed representation of any
integer function which supports its application and the application of
its iterated variants in small time.

\begin{corollary}\label{cor:fun}
  There is a representation of a function $f:[n]\rightarrow[n]$ that
  uses $n (1+ \lceil\lg\nRuns\rceil + H^*(T)) +o(n\lg n) $ bits to
  support $f^k(i)$ in $\bigo(\log\nRuns+|f^k(i)|)$ time, for any
  integer $k$ and for any $i\in[n]$, where $T$ is the forest
  representing the injective part of the function, and $\nRuns$ is the
  number of runs in the bijective part of the function.
\end{corollary}

\section{Conclusion}
\label{sec:conclusion}

Bentley and Yao~\cite{anAlmostOptimalAlgorithmForUnboundedSearching},
when introducing a family of search algorithms adaptive to the
position of the element searched (aka the ``unbounded search''
problem), did so through the definition of a family of adaptive codes
for unbounded integers, hence proving that the link between algorithms
and encodings was not limited to the complexity lower bounds suggested
by information theory.
%

In this paper, we have considered the relation between the difficulty
measures of adaptive sorting algorithms and some measures of
``entropy'' for compression techniques on permutations.
In particular, we have shown that some concepts originally defined for
adaptive sorting algorithms, such as runs and shuffled upsequences,
are useful in terms of the compression of permutations; and conversely,
that concepts originally defined for data compression, such as the
entropy of the sets of sizes of runs, are a useful addition to the
set of difficulty measures that one can consider in the study of
adaptive algorithms.

It is easy to generalize our results on runs and strict
runs to take advantage of permutations which are a mix of up and down
runs or strict runs
(e.g. $(\mathit{1},\mathit{3},\mathit{5},\mathit{7},\mathit{9},\mathbf{10},\mathbf{8},\mathbf{6},\mathbf{4},\mathbf{2})$,
with only a linear extra computational and/or space cost.
The generalization of our results on shuffled upsequences to 
SMS~\cite{sortingShuffledMonotoneSequences}, permutations
containing mixes of subsequences sorted in increasing and decreasing
orders
(e.g. $(\mathit{1},\mathbf{10},\mathit{2},\mathbf{9},\mathit{3},\mathbf{8},\mathit{4},\mathbf{7},\mathit{5},\mathbf{6})$)
is sligthly more problematic, because it is NP hard to optimally
decompose a permutation into such
subsequences~\cite{partitioningPermutationsIntoIncreasingAndDecreasingSubsequences},
but any approximation scheme~\cite{sortingShuffledMonotoneSequences} would yield a good encoding.

Refer to the associated technical
report~\cite{compressedRepresentationsOfPermutationsAndApplicationsTR}
for a longer version of this paper, in particular including all
the proofs.




\begin{thebibliography}{10}

\bibitem{ANS06}
D.~Arroyuelo, G.~Navarro, and K.~Sadakane.
\newblock Reducing the space requirement of {LZ}-index.
\newblock In {\em Proc. 17th CPM}, LNCS 4009, pages 319--330, 2006.

\bibitem{BY04}
R.~Baeza-Yates.
\newblock A fast set intersection algorithm for sorted sequences.
\newblock In {\em Proc. 15th CPM}, LNCS 3109, pages 400--408, 2004.

\bibitem{BYRN99}
R.~Baeza-Yates and B.~Ribeiro.
\newblock {\em Modern Information Retrieval}.
\newblock Addison-Wesley, 1999.

\bibitem{adaptiveSearchingInSuccinctlyEncodedBinaryRelationsAndTreeStructuredDocumentsTCS}
J.~Barbay, A.~Golynski, J.~I. Munro, and S.~S. Rao.
\newblock Adaptive searching in succinctly encoded binary relations and
  tree-structured documents.
\newblock {\em Theor. Comp. Sci.}, 2007.

\bibitem{succinctIndexesForStringsBinaryRelationsAndMultiLabeledTrees}
J.~Barbay, M.~He, J.~I. Munro, and S.~S. Rao.
\newblock Succinct indexes for strings, binary relations and multi-labeled
  trees.
\newblock In {\em Proc. 18th SODA}, pages 680--689, 2007.

\bibitem{fasterAdaptiveSetIntersectionsForTextSearching}
J.~Barbay, A.~L{\'o}pez-Ortiz, and T.~Lu.
\newblock Faster adaptive set intersections for text searching.
\newblock In {\em Proc. 5th WEA}, LNCS 4007, pages 146--157, 2006.

\bibitem{compressedRepresentationsOfPermutationsAndApplicationsTR}
J.~Barbay and G.~Navarro.
\newblock Compressed representations of permutations, and applications.
\newblock Technical Report TR/DCC-2008-18, Department of Computer Science
  (DCC), University of Chile, December 2008.
\newblock \url{http://www.dcc.uchile.cl/TR/2008/TR_DCC-2008-018.pdf}.

\bibitem{anAlmostOptimalAlgorithmForUnboundedSearching}
J.~L. Bentley and A.~C.-C. Yao.
\newblock An almost optimal algorithm for unbounded searching.
\newblock {\em Inf. Proc. Lett.}, 5(3):82--87, 1976.

\bibitem{BFLN08}
N.~Brisaboa, A.~Fari{\~n}a, S.~Ladra, and G.~Navarro.
\newblock Reorganizing compressed text.
\newblock In {\em Proc. 31st SIGIR}, pages 139--146, 2008.

\bibitem{Cla96}
D.~Clark.
\newblock {\em Compact Pat Trees}.
\newblock PhD thesis, University of Waterloo, Canada, 1996.

\bibitem{CN08}
F.~Claude and G.~Navarro.
\newblock Practical rank/select queries over arbitrary sequences.
\newblock In {\em Proc. 15th SPIRE}, LNCS 5280, pages 176--187, 2008.

\bibitem{cook}
C.~Cool and D.~Kim.
\newblock Best sorting algorithm for nearly sorted lists.
\newblock {\em Comm. ACM}, 23:620--624, 1980.

\bibitem{CM07}
J.~Culpepper and A.~Moffat.
\newblock Compact set representation for information retrieval.
\newblock In {\em Proc. 14th SPIRE}, pages 137--148, 2007.

\bibitem{estivillcastro92survey}
V.~Estivill-Castro and D.~Wood.
\newblock A survey of adaptive sorting algorithms.
\newblock {\em ACM Comp. Surv.}, 24(4):441--476, 1992.

\bibitem{FMMN07}
P.~Ferragina, G.~Manzini, V.~M{\"a}kinen, and G.~Navarro.
\newblock Compressed representations of sequences and full-text indexes.
\newblock {\em ACM Trans. on Algorithms (TALG)}, 3(2):article 20, 2007.

\bibitem{onComputingTheLengthOfLongestIncreasingSubsequences}
M.~L. Fredman.
\newblock On computing the length of longest increasing subsequences.
\newblock {\em Discrete Math.}, 11:29--35, 1975.

\bibitem{Gol06}
A.~Golynski.
\newblock Optimal lower bounds for rank and select indexes.
\newblock In {\em Proc. 33th ICALP}, LNCS 4051, pages 370--381, 2006.

\bibitem{rankSelectOperationsOnLargeAlphabets}
A.~Golynski, J.~I. Munro, and S.~S. Rao.
\newblock Rank/select operations on large alphabets: a tool for text indexing.
\newblock In {\em Proc. 17th SODA}, pages 368--373, 2006.

\bibitem{GGV03}
R.~Grossi, A.~Gupta, and J.~Vitter.
\newblock High-order entropy-compressed text indexes.
\newblock In {\em Proc. 14th SODA}, pages 841--850, 2003.

\bibitem{GV06}
R.~Grossi and J.~Vitter.
\newblock Compressed suffix arrays and suffix trees with applications to text
  indexing and string matching.
\newblock {\em SIAM J. on Computing}, 35(2):378--407, 2006.

\bibitem{guibas}
L.~Guibas, E.~McCreight, M.~Plass, and J.~Roberts.
\newblock A new representation of linear lists.
\newblock In {\em Proc. 9th STOC}, pages 49--60, 1977.

\bibitem{Hea78}
H.~Heaps.
\newblock {\em Information Retrieval - Computational and Theoretical Aspects}.
\newblock Academic Press, NY, 1978.

\bibitem{HT71}
T.~Hu and A.~Tucker.
\newblock Optimal computer-search trees and variable-length alphabetic codes.
\newblock {\em SIAM J. of Applied Mathematics}, 21:514--532, 1971.

\bibitem{Huf52}
D.~Huffman.
\newblock A method for the construction of minimum-redundancy codes.
\newblock {\em Proceedings of the I.R.E.}, 40(9):1090--1101, 1952.

\bibitem{ultraSuccinctRepresentationOfORderedTrees}
J.~Jansson, K.~Sadakane, and W.-K. Sung.
\newblock Ultra-succinct representation of ordered trees.
\newblock In {\em Proc. 18th SODA}, pages 575--584, 2007.

\bibitem{partitioningPermutationsIntoIncreasingAndDecreasingSubsequences}
A.~E. K\'{e}zdy, H.~S. Snevily, and C.~Wang.
\newblock Partitioning permutations into increasing and decreasing
  subsequences.
\newblock {\em J. Comb. Theory Ser. A}, 73(2):353--359, 1996.

\bibitem{theArtOfComputerProgrammingVol3}
D.~E. Knuth.
\newblock {\em The Art of Computer Programming, Volume 3: Sorting and
  Searching}.
\newblock Addison-Wesley, 2nd edition, 1998.

\bibitem{sortingShuffledMonotoneSequences}
C.~Levcopoulos and O.~Petersson.
\newblock Sorting shuffled monotone sequences.
\newblock {\em Inf. Comp.}, 112(1):37--50, 1994.

\bibitem{MN05}
V.~M{\"a}kinen and G.~Navarro.
\newblock Succinct suffix arrays based on run-length encoding.
\newblock {\em Nordic J. of Computing}, 12(1):40--66, 2005.

\bibitem{MNspire07}
V.~M{\"a}kinen and G.~Navarro.
\newblock Implicit compression boosting with applications to self-indexing.
\newblock In {\em Proc. 14th SPIRE}, LNCS 4726, pages 214--226, 2007.

\bibitem{MN07}
V.~M{\"a}kinen and G.~Navarro.
\newblock Rank and select revisited and extended.
\newblock {\em Theor. Comp. Sci.}, 387(3):332--347, 2007.

\bibitem{measuresOfPresortednessAndOptimalSortingAlgorithms}
H.~Mannila.
\newblock Measures of presortedness and optimal sorting algorithms.
\newblock In {\em IEEE Trans. Comput.}, volume~34, pages 318--325, 1985.

\bibitem{Man01}
G.~Manzini.
\newblock An analysis of the {B}urrows-{W}heeler transform.
\newblock {\em J. of the ACM}, 48(3):407--430, 2001.

\bibitem{mehlhorn}
K.~Mehlhorn.
\newblock Sorting presorted files.
\newblock In {\em Proc. 4th GI-Conference on Theoretical Computer Science},
  LNCS 67, pages 199--212, 1979.

\bibitem{MNZBY00}
E.~Moura, G.~Navarro, N.~Ziviani, and R.~Baeza-Yates.
\newblock Fast and flexible word searching on compressed text.
\newblock {\em ACM Trans. on Information Systems (TOIS)}, 18(2):113--139, 2000.

\bibitem{Mun96}
I.~Munro.
\newblock Tables.
\newblock In {\em Proc. 16th FSTTCS}, LNCS 1180, pages 37--42, 1996.

\bibitem{Munro03}
J.~I. Munro, R.~Raman, V.~Raman, and S.~S. Rao.
\newblock Succinct representations of permutations.
\newblock In {\em Proc. 30th ICALP}, LNCS 2719, pages 345--356, 2003.

\bibitem{representingFunctions}
J.~I. Munro and S.~S. Rao.
\newblock Succinct representations of functions.
\newblock In {\em Proc. 31st ICALP}, LNCS 3142, pages 1006--1015, 2004.

\bibitem{NM07}
G.~Navarro and V.~M{\"a}kinen.
\newblock Compressed full-text indexes.
\newblock {\em ACM Comp. Surv.}, 39(1):article 2, 2007.

\bibitem{RRR02}
R.~Raman, V.~Raman, and S.~Rao.
\newblock Succinct indexable dictionaries with applications to encoding $k$-ary
  trees and multisets.
\newblock In {\em Proc. 13th SODA}, pages 233--242, 2002.

\bibitem{Sad03}
K.~Sadakane.
\newblock New text indexing functionalities of the compressed suffix arrays.
\newblock {\em J. of Algorithms}, 48(2):294--313, 2003.

\bibitem{ST07}
P.~Sanders and F.~Transier.
\newblock Intersection in integer inverted indices.
\newblock In {\em Proc. 9th ALENEX}, 2007.

\bibitem{MNSV08}
J.~Sir{\'e}n, N.~V{\"a}lim{\"a}ki, V.~M{\"a}kinen, and G.~Navarro.
\newblock Run-length compressed indexes are superior for highly repetitive
  sequence collections.
\newblock In {\em Proc. 15th SPIRE}, LNCS 5280, pages 164--175, 2008.

\bibitem{encroachingListsAsAMeasureOfPresortedness}
S.~S. Skiena.
\newblock Encroaching lists as a measure of presortedness.
\newblock {\em BIT}, 28(4):775--784, 1988.

\bibitem{ST85}
D.~Sleator and R.~Tarjan.
\newblock Self-adjusting binary search trees.
\newblock {\em J. of the ACM}, 32(3):652--686, 1985.

\bibitem{WMB99}
I.~Witten, A.~Moffat, and T.~Bell.
\newblock {\em Managing Gigabytes}.
\newblock Morgan Kaufmann, 2nd edition, 1999.

\end{thebibliography}
\end{document}